\documentclass[preprint,12pt]{elsarticle}
\usepackage{graphics}
\usepackage{graphicx}
\usepackage{amssymb}
\usepackage{amsthm}
\usepackage{natbib}
\journal{Methods}
\begin{document}
\begin{frontmatter}

\title{Signalling entropy: a novel network-theoretical framework for systems analysis and interpretation of functional omic data}
\author[f1,f2]{Andrew E. Teschendorff\corref{l1}}
\author[f3]{Peter Sollich}
\author[f3]{Reimer Kuehn}
\address[f1]{CAS-MPG Partner Institute for Computational Biology, Chinese Academy of Sciences, Shanghai Institute for Biological Sciences, 320 Yue Yang Road, Shanghai 200031, China.}
\address[f2]{Statistical Cancer Genomics, Paul O'Gorman Building, UCL Cancer
Institute, University College London, London WC1E 6BT, UK.}
\address[f3]{Department of Mathematics, King's College London, London WC2R 2LS, UK.}
\cortext[l1]{a.teschendorff@ucl.ac.uk}

\begin{abstract}
A key challenge in systems biology is the elucidation of the underlying principles, or fundamental laws, which determine the cellular phenotype. Understanding how these fundamental principles are altered in diseases like cancer is important for translating basic scientific knowledge into clinical advances. While significant progress is being made, with the identification of novel drug targets and treatments by means of systems biological methods, our fundamental systems level understanding of why certain treatments succeed and others fail is still lacking. We here advocate a novel methodological framework for systems analysis and interpretation of molecular omic data, which is based on statistical mechanical principles. Specifically, we propose the notion of cellular signalling entropy (or uncertainty), as a novel means of analysing and interpreting omic data, and more fundamentally, as a means of elucidating systems-level principles underlying basic biology and disease. We describe the power of signalling entropy to discriminate cells  according to differentiation potential and cancer status. We further argue the case for an empirical cellular entropy-robustness correlation theorem and demonstrate its existence in cancer cell line drug sensitivity data. Specifically, we find that high signalling entropy correlates with drug resistance and further describe how entropy could be used to identify the achilles heels of cancer cells. In summary, signalling entropy is a deep and powerful concept, based on rigorous statistical mechanical principles, which, with improved data quality and coverage, will allow a much deeper understanding of the systems biological principles underlying normal and disease physiology.
\end{abstract}


\end{frontmatter}

\section{Introduction}
Recent advances in biotechnology are allowing us to measure cellular properties at an unprecedented level of detail \cite{Hood2013}. For instance, it is now possible to routinely measure various molecular entities (e.g. DNA methylation, mRNA and protein expression, SNPs) genome-wide in hundreds if not thousands of cellular specimens \cite{TCGAovc2011}. In addition, other molecular data detailing interactions between proteins or between transcription factors and regulatory DNA elements are growing at a rapid pace \cite{Hood2013}. All these types of data are now widely referred to collectively as ``omic'' data. The complexity and high-dimensional nature of this omic data presents a daunting challenge to those wishing to analyse and interpret the data \cite{Hood2013}. The difficulty of analysing omic data is further compounded by the inherent complexity of cellular systems. Cells are prime examples of organized complex systems, capable of highly stable and predictable behaviour, yet an understanding of how this deterministic behaviour emerges from what is a highly complex and probabilistic pattern of dynamic interactions between numerous intra and extracellular components, still eludes us \cite{Hood2013}. Thus, elucidating the systems-biological laws or principles dictating cellular phenotypes is also key for an improved analysis and interpretation of omic data. Furthermore, important biological phenomena such as cellular differentiation are fundamentally altered in diseases like cancer \cite{Hanahan2011}. Hence, an attempt to understand how cellular properties emerge at a systems level from the properties seen at the individual gene level is not only a key endeavour for the systems biology community, but also for those wanting to translate basic insights into effective medical advances \cite{Barabasi2011,Creixell2012}.\\
It is now well accepted that at a fundamental level most biological systems are best modeled in terms of spatial interactions between specific entities (e.g. neurons in the case of the brain), which may or may not be dynamically changing in time \cite{Weng1999,Barabasi2013}. It therefore seems natural to also use the mathematical and physical framework of networks to help us analyse and interpret omic data at a systems level \cite{Teschendorff2010bmc,West2012}. Indeed, the cellular phenotype is determined to a large extent by the precise pattern of molecular interactions taking place in the cell, i.e. a molecular interaction network \cite{Barabasi2004}. Although this network has spatial and dynamic dimensions which at present remain largely unexplored due to technological or logistical limitations, there is already a growing number of examples where network-based analysis strategies have been instrumental \cite{Robin2013,Erler2012,Lee2012}. For instance, a deeper understanding of why sustained treatment with EGFR inhibitors can lead to dramatic sensitization of cancer cell lines to cytotoxic agents was possible thanks to a systems approach \cite{Lee2012}. Another study used reverse engineering network approaches to identify and validate drug targets in glioblastoma multiforme, to be further tested in clinical trials \cite{Carro2010}. What is key to appreciate here is that these successes have been achieved in spite of noisy and incomplete data, suggesting that there are simple yet deep systems biological principles underlying cellular biology that we can already probe and exploit with current technology and data. Thus, with future improvements in data quality and coverage, network-based analysis frameworks will play an ever increasing and important role in systems biology, specially at the level of systems analysis and interpretation \cite{Barabasi2011}. Therefore, it is also imperative to develop novel, more powerful, network-theoretical methods for the systems analysis of omic data.\\
In adopting a network's perspective for the analysis and interpretation of omic data, there are in principle two different (but not mutually exclusive) approaches one can take. One possibility is to infer (i.e. reverse engineer) the networks from genome-wide data \cite{Lefebvre2012}. Most of these applications have done this in the context of gene expression data, with the earliest approaches using clustering or co-expression to postulate gene interdependencies \cite{Segal2004}. Partial correlations and Graphical Gaussian Models have proved useful as a means of further refining correlation networks by allowing one to infer the more likely direct interactions while simultaneously also filtering out those which are more likely to be indirect \cite{Opgen-Rhein2007}. These methods remain popular and continue to be studied and improved upon \cite{Barzel2013,Feizi2013}. Other methods have drawn on advanced concepts from information theory, for instance ARACNe (``Algorithm for the Reconstruction of Accurate Cellular Networks'') has been shown to be successful in infering regulatory networks in B-cells \cite{Lefebvre2012}.\\
In stark contrast to reverse engineering methods, another class of algorithms have used structural biological networks from the outset, using these as scaffolds to integrate with omic data. Specifically, by using a structural network one can sparsify the correlation networks inferred from reverse-engineering approaches, thus providing another means of filtering out correlations that are more likely to be indirect \cite{West2012}. Besides, integration with a structural network automatically provides an improved framework for biological interpretation \cite{Barabasi2004,Chuang2007,Dutkowski2011,Mitra2013}. The structural networks themselves are typically derived from large databases, which detail literature curated experimentally verified interactions, including interactions derived from Yeast 2 Hybrid screens (Y2H) \cite{Prasad2009}. The main example is that of protein protein interaction (PPI) maps, which have been generated using a number of different complementary experimental and in-silico approaches, and merging these maps from these different sources together has been shown to be an effective approach in generating more comprehensive high-confidence interaction networks \cite{Cerami2011,Vidal2011}. PPI networks have been used mainly as a means of integrating and analysing gene expression data (see e.g. \cite{Chuang2007,Ulitsky2007,Taylor2009}). More recently, this approach has also been successfully applied in the DNA methylation context, for instance it has been shown that epigenetic changes associated with age often target specific gene modules and signalling pathways \cite{West2013}.\\
Another class of methods that have used structural networks, PPIs in particular, have integrated them with gene expression data to define an approximation to the underlying signaling dynamics on the network, thus allowing more in-depth exploration of the interplay between network topology and gene expression. Typically, these studies have used the notion of random walks on weighted graphs where the weights are constructed from differential expression statistics, and where the aim is to identify nodes (genes) in the network which may be important in dictating the signaling flows within the pathological state. For instance, among these random walk methods is NetRank, a modification of the Google PageRank algorithm, which was able to identify novel, robust, network based biomarkers for survival time in various cancers \cite{Roy2012,Winter2012}. Other random walk based approaches, aimed at identifying causal drivers of specific phenotypes (e.g. expression or cancer), have modeled signal transduction between genes in the causal and phenotypic layers as flows in an electric circuit diagram, an elegant formulation capable of not only identifying the likely causal genes but also of tracing the key pathways of information flow or dysregulation \cite{Suthram2008,Kim2011}. Random walk theory has also been employed in the development of differential network methodologies. An example is NetWalk \cite{Komurov2010}, which is similar to NetRank but allows differential signaling fluxes to be inferred. This approach was successful in identifying and validating the glucose metabolic pathway as a key determinant of lapatinib resistance in ERBB2 positive breast cancer patients \cite{Komurov2012}.
\\
Another important concept to have emerged recently is that of network rewiring \cite{Ideker2010,Califano2011,Ideker2012}. This refers to the changes in interaction patterns that accompany changes in the cellular phenotype. Network rewiring embodies the concept that it is the changes in the interaction patterns, and not just the changes in absolute gene expression or protein activity, that are the main determinants of the cellular phenotype. That network rewiring may be key to understanding cellular phenotypes was most convincingly demonstrated in a differential epistasis mapping study conducted in yeast cells exposed to a DNA damaging agent \cite{Ideker2010}. Specifically, what this study demonstrated is that responses to perturbations or cellular stresses are best understood in terms of the specific rewiring of protein complexes and functional modules. Thus, this conceptual framework of network rewiring may apply equally well to the genetic perturbations and cellular stresses underlying disease pathologies like cancer.\\
In this article we advocate a network-theoretical framework based on statistical mechanical principles and more specifically on the notion of signalling entropy \cite{Teschendorff2010bmc,West2012}. This theoretical framework integrates gene expression (but in principle also other functional data) with a PPI network, merging existing concepts such as signaling dynamics (i.e. random walks) and network rewiring with that of signalling entropy. In previous work we have shown how signalling entropy (i) provides a proxy to the elevation in Waddington's epigenetic landscape, correlating with a cell's differentiation potential \cite{Banerji2013}, (ii) how it can be used to identify signaling pathways and nodes important in differentiation and cancer \cite{Banerji2013,Teschendorff2010bmc,West2012}, and (iii) how it predicts two cancer system-omic hallmarks: (a) cancer is characterised by an increase in signalling entropy and (b) local signaling entropy changes anti-correlate with differential gene expression \cite{Teschendorff2010bmc,West2012}. Here, we present and unify the different signaling entropy measures used previously and further explore a novel application of signalling entropy to understanding drug sensitivity profiles in cancer cell lines. Specifically, we first use simulated data to justify the existence of an entropy-robustness theorem, and subsequently provide empirical evidence for this theorem by demonstrating that increases in local signalling entropy correlate with drug resistance (robustness). We further show the importance of network topology in dictating the signalling entropy changes underlying drug response. In addition, we provide R-functions implementing the entropy rate calculation and ranking of genes according to differential entropy, all freely available from {\it sourceforge.net/projects/signalentropy/files/} .

\section{Materials and Methods}

\subsection{Basic rationale and motivation for Signalling Entropy: understanding systems biology through uncertainty}
Loosely defined, entropy of a system, refers to a measure of the disorder, randomness or uncertainty of processes underlying the system's state. In the context of a single cell, signalling entropy will reflect the amount of overall disorder, randomness or uncertainty in how information, i.e. signaling, is passed on in the molecular interaction network. At a fundamental level, all signaling in a cell is probabilistic, determined in part by the relative cellular concentrations of the interacting molecules. Hence, in discussing cellular signalling entropy, it is useful to picture a molecular interaction network in which edges represent possible interactions and with the edge weights reflecting the relative probabilities of interaction ({\bf Fig.1A}). Thus, an interaction network in a high-entropy state is characterised by signaling interaction probabilities that are all fairly similar in value, whereas a low-entropy state will be characterised by specific signalling interactions possessing much higher weights ({\bf Fig.1A}).

\begin{figure}[ht]
\begin{center}
\includegraphics[scale=0.6]{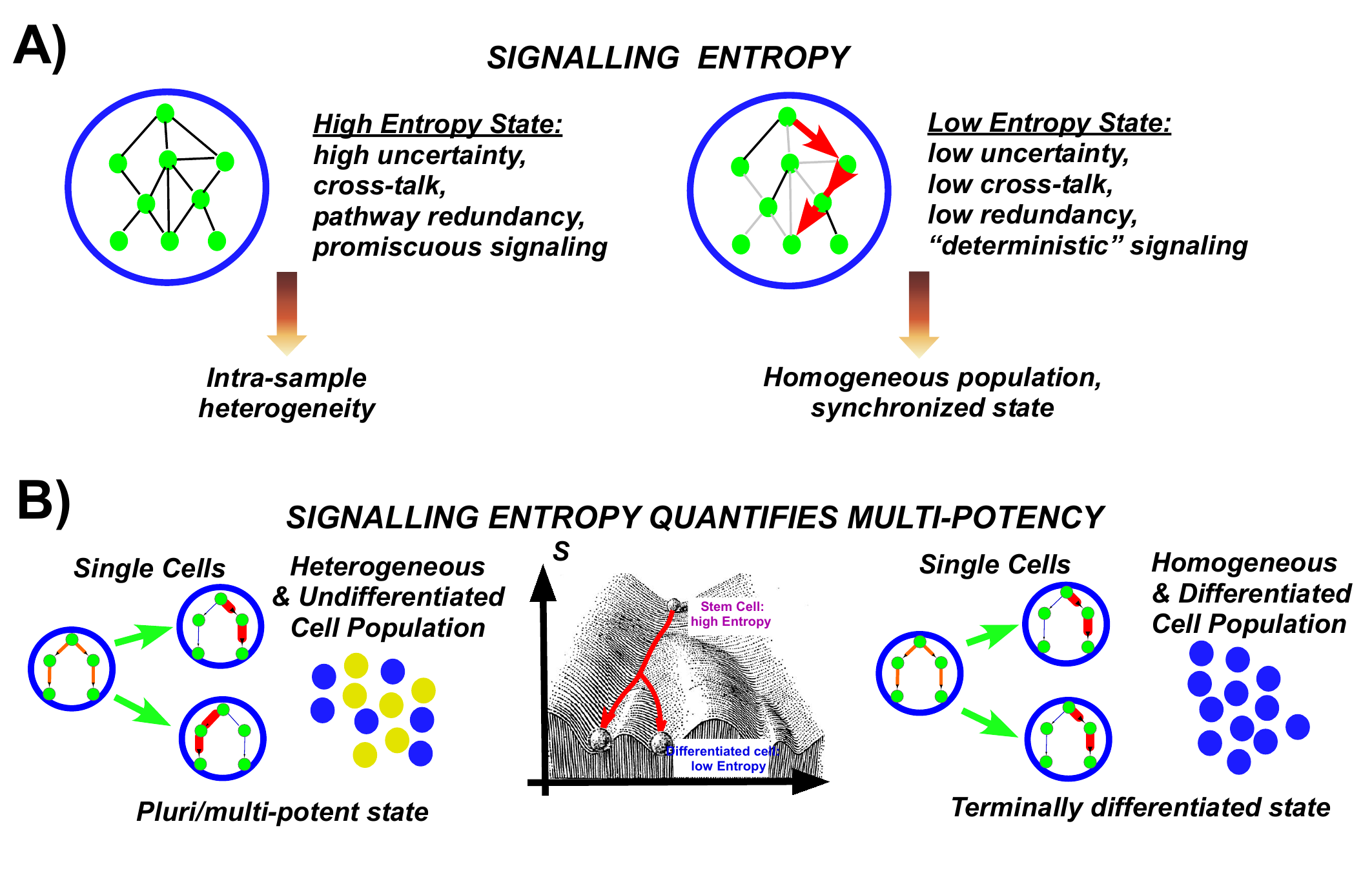}
%
%
\caption{{\bf Signalling Entropy: understanding systems biology through uncertainty. A)} A caricature model of a cellular interaction network with edge widths/color indicating the relative probabilities of interaction. On the left and right, we depict states of high and low signalling entropy, respectively. At the cellular population level, this translates into samples of high and low intra-sample heterogeneity, respectively. {\bf B)} Signalling entropy correlates with pluripotency as demonstrated in our previous work (Banerji et al 2013). The pluripotent state is a highly promiscuous signaling state, generating high intra-sample heterogeneity, and allowing the stem cell population to differentiate into any cell type. In contrast, in a terminally differentiated state, signaling is almost deterministic, reflecting activation of very specific pathways in the majority of cells, leading to a highly homogeneous and differentiated cell population. Thus, signalling entropy defines the height in Waddington's differentiation landscape. 
}
\label{fig:1}       
\end{center}
\end{figure}

Why would this type of signalling entropy, loosely defined as the amount of uncertainty in the signaling interaction patterns, be useful to systems biology? One way to think of signalling entropy is as representing signaling promiscuity, which has been proposed as a key systems feature underlying the pluripotent or multipotent capacity of cells ({\bf Fig.1B}) \cite{Chang2008nat,Macarthur2013,Furusawa2009,Furusawa2012}. Indeed, it has been suggested that pluripotency is an intrinsic statistical mechanical property, best defined at the cellular population level \cite{Macarthur2013}. Specifically, it has been demonstrated that pluripotent stem cells exhibit remarkable heterogeneity in gene expression levels, including well known stem cell markers such as {\it NANOG} \cite{Macarthur2013}. It is also well known that a large number of genes, many encoding transcription factors, exhibit low-levels of expression in stem cells, yet simultaneously are being kept in a poised chromatin state, allowing immediate activation if this were required \cite{Lee2006pcgt}. Thus, in a pluripotent stem cell like state, signal transduction is in a highly egalitarian and, thus, promiscuous state, i.e. a state of high signalling entropy. Conversely, differentiation leads, by necessity, to activation of specific transcription factors and pathways and thus to a lowering in the uncertainty of signaling patterns, and thus to a lowering of entropy. We recently demonstrated, using gene expression profiles of over 800 samples, comprising cells of all major stages of differentiation, including human embryonic stem cells (hESCs), induced pluripotent stem cells (iPSCs), multipotent cell types (e.g. hematopoietic stem cells (HSCs)), and terminally differentiated cells within these respective lineages, that signalling entropy not only correlates with differentiation potential but that it provides a highly quantitative measure of potency \cite{Banerji2013}. Indeed, we showed that signalling entropy provides a reasonably good approximation to the energy potential in Waddington's epigenetic landscape \cite{Banerji2013}.\\
Here we decided to explore the concept of signalling entropy in relation to cellular robustness and specifically to drug resistance in cancer. That signalling entropy may be informative of a cell's robustness is a proposal that stems from a general (but unproven) theorem, first proposed by Manke and Demetrius \cite{Demetrius2004,Demetrius2005,Manke2006}: namely, that a system's entropy and robustness are correlated. Mathematically, this can be expressed as $\Delta S\Delta R>0$, which states that a positive change in a system's entropy (i.e. $\Delta S>0$) must lead to an increase in robustness ($\Delta R>0$). Now, cells are fairly robust entities, having evolved the capacity to buffer the intra-and-extracellular stresses and noise which they are constantly exposed to \cite{Stelling2004,Chang2008nat}. Much of this overall stability and robustness likely stems from the topological features of the underlying signaling and regulatory networks, for instance features such as scale-freeness and hierarchical modularity, which are thought to have emerged through natural evolutionary processes such as gene duplication \cite{Barabasi2004,Han2004,LiJingJing2010}. However, another key feature which contributes to cellular robustness is cross-talk and signalling pathway redundancy \cite{Stelling2004}. Pathway redundancy refers to a situation where a cell has the choice of transmitting signals via two or more possible routes. In the language of statistical mechanics, this corresponds to a state of high uncertainty (or entropy) in signaling. High signalling entropy could thus underpin a cell's robustness to perturbations, suggesting that a cell's entropy and robustness may indeed be correlated ({\bf Fig.2A}). Consistent with this, pathway redundancy is also well recognized to be a key feature underlying drug resistance of cancer cells \cite{Engelman2008}.

\begin{figure}[ht]
\begin{center}
\includegraphics[scale=0.5]{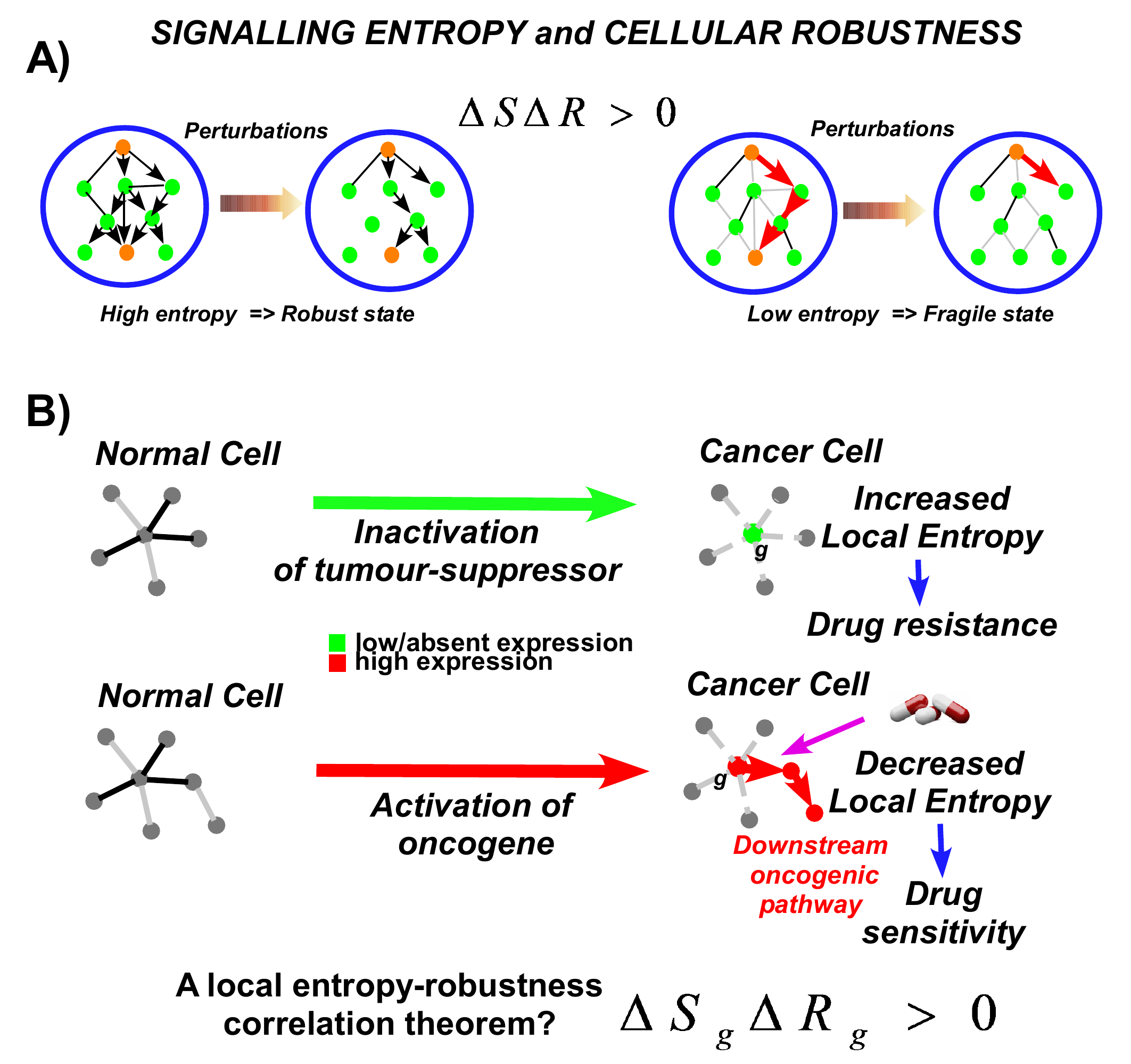}
%
%
\caption{{\bf Signalling entropy and cellular robustness:}  {\bf A)} Signalling entropy ought to correlate with cellular robustness. The inequality encapsulates this information by stating that a decrease in signalling entropy (i.e. if $\Delta S <0$), then the system's robustness $R$ must also decrease, i.e. $\Delta R<0$, so that the product $\Delta S\Delta R>0$. Observe how in the low entropy state, random removal of edges through e.g. inactivating mutations, can lead to deactivation of a key signaling pathway connecting a given start and end nodes (shown in orange). In the high entropy state, the same perturbations do not prevent signal transduction between the two orange nodes. {\bf B)} Depicted are the effects of two major forms of cancer perturbation. In the upper panel, inactivation (typically of tumour suppressors), leads to underexpression and a corresponding loss of correlations/interactions with neighbors in the PPI network. This is tantamount to a state of increased entropy and drug intervention is unlikely to be effective. In the lower panel, we depict the case of an oncogene, which is overexpressed in cancer. This overexpression leads to activation of a specific oncogenic pathway which results in oncogene addiction and increased sensitivity to targeted drug intervention. Thus, local signalling entropy and robustness (as determined by response to a drug), may also correlate locally. 
}
\label{fig:2}       
\end{center}
\end{figure}

 Further supporting the notion that entropy and robustness may be correlated, we previously showed that (i) cancer is characterised by a global increase in signalling entropy compared to its respective normal tissue \cite{West2012}, in line with the observation that cancer cells are specially robust to generic perturbations, and (ii) that a gene's differential entropy between normal and cancer tissue is generally speaking anticorrelated with its differential expression \cite{West2012}, consistent with the view that cancer cells are specially sensitive to drug interventions that target overexpressed oncogenes, a phenomenon known as oncogene addiction ({\bf Fig.2B}) \cite{Hanahan2011}. Interestingly, the observations that differential entropy and differential expression are anti-correlated and that cancer is characterised globally by an increase in entropy \cite{West2012}, are also consistent with the prevalent view that most driver mutations are indeed inactivating, targeting tumor suppressor genes \cite{Sjoblom2006,Wood2007,Vogelstein2013}. Hence, based on all of these observations and insights, we posited that signalling entropy could also prove useful as a means of predicting drug resistance in cancer cells.

\subsection{The Signalling Entropy Method: definitions and construction}
Briefly, we review the definitions and construction of signalling entropy as used in our previous studies \cite{Teschendorff2010bmc,West2012,Banerji2013}. The construction relies on a comprehensive and high-confidence PPI network which is integrated with gene expression data \cite{Cerami2011,West2012} (see {\bf Appendix A}). Briefly, the PPI is used as a scaffold, and edge weights are constructed from the gene expression data to approximate the interaction or signaling probabilities between the corresponding proteins in the PPI. Thus, the  gene expression data provides the biological context in which to modulate the PPI interaction probabilities. To compute signalling entropy requires the estimation of a stochastic matrix, reflecting these interaction probablities over the network.\\ 
The construction of the stochastic matrix can proceed in two different ways. In the earliest studies we used a construction which was done at the level of phenotypes \cite{Teschendorff2010bmc,West2012}. Under this model, edge weights $w_{ij}$ between proteins $i$ and $j$ were constructed from the correlation between the expression levels of the corresponding genes $i$ and $j$, as assessed over independent samples all representing the same phenotype. Estimating the correlations over independent samples, all within the same phenotype, can be viewed as representing a multifactorial perturbation experiment, with e.g. genetic differences between individuals mimicking specific perturbations, and thus allowing influences to be inferred. Thus, this approach hinges on the assumption that highly correlated genes, whose coding proteins interact, are more likely to be interacting in the given phenotype than two genes which do not correlate. The use of a PPI as a scaffold is important to filter out significant correlations which only result from indirect influences. The correlations themselves can be defined in many different ways, for instance, using simple Pearson correlations or non-linear measures such as Mutual Information. For example, one way to define the weights is as $w_{ij}=\frac{1}{2}(1+c_{ij})$ with $c_{ij}$ describing the Pearson correlation coefficient between genes $i$ and $j$. This definition guarantees positivity, but also treats positive and negative correlations differently, which makes biological sense because activating and inhibitory interactions normally have completely distinct consequences on downstream signalling. Thus, the above weight scheme treats zero or insignificant correlations as intermediate, consistent with the view that an absent interaction is neither inhibitory nor activating. However, other choices of weights are possible: e.g. $w_{ij}=|c_{ij}|$, which treats negative and positive correlations on an equal footing. Once edge weights are defined as above, these are then normalised to define the stochastic matrix $p_{ij}$ over the network, $$p_{ij}=\frac{w_{ij}}{\sum_{k \in \mathcal{N}_{i}}w_{ik}},$$ with $\mathcal{N}_i$ denoting the PPI neighbors of gene $i$. Thus, $p_{ij}$ is the probability of interaction between genes $i$ and $j$, and as required, $\sum_j{p_{ij}}=1$. However, there is no requirement for $p_{ij}$ to be doubly stochastic, i.e. $P$ is in general not a symmetric matrix. Hence, edges are bi-directional with different weights labeling the probability of signal transduction from $i$ to $j$, and that from $j$ to $i$ ($p_{ij}\neq p_{ji}$).\\
An alternative to the above construction of the stochastic matrix is to invoke the mass action principle, i.e. one now assumes that the probability of interaction in a given sample is proportional to the product of expression values of the corresponding genes in that sample \cite{Banerji2013}. Thus, the PPI is again used as a scaffold to only allow interactions supported by the PPI network, but the weights are defined using the mass action principle, as $$w_{ij}\propto E_iE_j$$ where $E_i$ denotes the normalised expression intensity value (or normalised RNA-Seq read count) of gene $i$. An important advantage of this construction is that the stochastic matrix is now sample specific, as the expression values are unique to each sample.\\
Given a stochastic matrix, $p_{ij}$, constructed using one of the two methods above, one can now define a local Shannon entropy for each gene $i$ as $$\tilde{S}_{i}=-\frac{1}{\log{k_{i}}}\sum_{k \in \mathcal{N}_{i}}p_{ik}\log{p_{ik}},$$ where $k_i$ denotes the degree of gene $i$ in the PPI network. The normalisation is optional but ensures that this local Shannon entropy is normalised between 0 and 1. Clearly, if only one weight is non-zero, then the entropy attains its minimal value (0), representing a state of determinism or lowest uncertainty. Conversely, if all edges emanating from $i$ carry the same weight, the entropy is maximal (1), representing the case of highly promiscuous signaling. In principle, local Shannon entropies can thus be compared between phenotypes to identify genes where there are changes in the uncertainty of signaling. In the case where entropies are estimated at the phenotype level, jackknife approaches can be used to derive corresponding differential entropy statistics \cite{West2012}. Deriving statistics is important because node degree has a dramatic influence on the entropy variance, with high degree nodes exhibiting significantly lower variability in absolute terms, which nevertheless could be highly significant \cite{West2012}. In the case where entropies are estimated at the level of individual samples, ordinary statistical tests (e.g. rank sum tests) can be used to derive sensible P-values, assuming of course that enough samples exist within the phenotypes being compared.\\
In addition to the local entropy, it is also of interest to consider statistical properties of the distribution of local entropies, for instance their average. Comparing the average of local entropies between two phenotypes would correspond to a comparison of non-equilibrium entropy rates. To see this, consider the formal definition of the entropy rate $SR$ \cite{Latora1999,GomezGardenes2008},  i.e. $$SR=\sum_{i=1}^{n} \pi_{i} S_{i},$$ where $\pi_{i}$ is the stationary distribution (equivalently the left eigenvector with unit eigenvalue) of $P$ (i.e. $\pi P=\pi$), and where now $$S_i=-\sum_{k \in \mathcal{N}_{i}}p_{ik}\log{p_{ik}}.$$ Note that the entropy rate $SR$ is an equilibrium entropy since it involves the stationary distribution of the random walker. As such, the entropy rate also depends on the global topology of the network. Thus, the entropy rate is a weighted average of the local unnormalized entropies, with the weights specified by the stationary distribution. It follows that comparing the unweighted averages of local entropies reflects a comparison of a non-equilibrium entropy rate since the stationary distribution is never used.

\begin{figure}[ht]
\begin{center}
\includegraphics[scale=0.6]{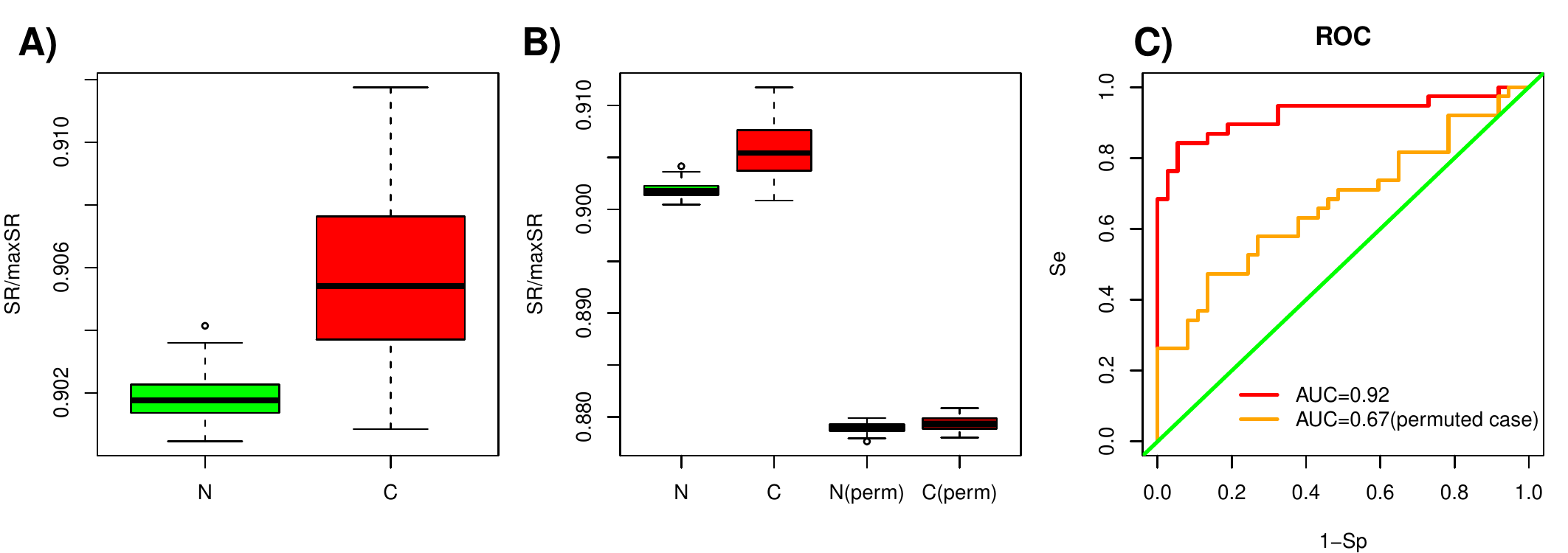}
%
%
\caption{{\bf Entropy rate in normal and cancer tissue:} {\bf A)} Boxplots of sample specific entropy rates comparing normal liver and liver cancer samples. Expression data set is the one used in \cite{Banerji2013}. {\bf B)} As A), but also shown are the sample specific entropy rates obtained by randomly permuting the gene expression values over the network. Note how the entropy rates for the normal and cancer states are significantly reduced upon permutation and are no longer highly discriminative between normal and cancer. {\bf C)} ROC curves and associated AUC normal-cancer discriminatory values for the unpermuted and permutated cases depicted in A) and B).
}
\label{fig:3}       
\end{center}
\end{figure}

\subsection{The importance of the integrated weighted network in estimating signalling entropy}
The entropy rate constructed using the mass action principle is sample specific. We previously demonstrated that this entropy rate was highly discriminative of the differentiation potential of samples within a developmental lineage, as well as being highly discriminative of normal and cancer tissue \cite{Banerji2013} ({\bf Fig.3A}). Since the entropy rate takes as input a PPI network and a sample's genome-wide expression profile, the significance of the resulting entropy values, as well as that of the difference between entropy rates, also needs to be assessed relative to a null distribution in which the putative information content between network and gene expression values is non-existent. In fact, since the weights in the network determine the entropy rate and these weights are dependent on both the specific network nodes and their respective gene expression profiles, it is natural to assess the significance of the entropy rate by ``destroying'' the mutual information between the network nodes and their gene expression profiles, for instance by randomising (i.e. permuting) the gene expression profiles over the network. Thus, under this randomisation, the topological properties of the network remain fixed, but the weights are redefined. Application of this randomisation procedure to the normal/cancer expression set considered previously \cite{Banerji2013} ({\bf Fig.3A}) shows that the discriminatory potential is significantly reduced upon permuting the gene expression values over the network ({\bf Fig.3B-C}). Importantly, we observe that the entropy rate is much higher in the normal and cancer states compared to the rates obtained upon randomisation of the gene expression profiles ({\bf Fig.3B}), indicating that both normal and cancer states are characterised by a higher level of signaling promiscuity compared to a network with random weights. That the discrimination between normal and cancer is significantly reduced in the randomly weighted network further demonstrates that there is substantial mutual information between the PPI network and the gene expression profiles, thus justifying the signalling entropy approach.

\subsection{Signalling entropy R-package: input, output and code availability}
A vignette/manual and user-friendly R-scripts that allow computation of the entropy rate is available at the following url: {\it sourceforge.net/projects/signalentropy/files/}. Here we briefly describe the salient aspects of this package:\\
{\it The input:} The main R-script provided ({\it CompSR}) takes as input a user-specified PPI network, and a genome-wide expression vector representing the gene expression profile of a sample. It is assumed that this has been generated using either Affy, Illumina or RNA-Sequencing. In principle one ought to use as gene expression value the normalised unlogged intensity (Affy/Illu) or RNA-seq count, since this is what should be proportional to the number of RNA transcripts in the sample. However, in practice we advise taking the log-transformed normalised value since the log-transformation provides a better compromise between proportionality and regularisation, i.e. some regularisation is advisable since the underlying kinetic reactions are also regular.\\
{\it The output:} The R-functions provided in the package then allow the user to estimate the global entropy rate for a given sample, as well as the local normalised entropies for each gene in the integrated network. If a phenotype is specified then genes can be ranked according to the correlation strength of their local entropy to the phenotype. Thus, the signalling entropy method allows us to assess (i) if the overall levels of signalling promiscuity is different between phenotypes, which could be important, for instance, to compare the pluripotent capacity of iPSCs generated via different protocols or to assess the stem cell nature of putative cancer stem cells \cite{Banerji2013}, and (ii) to rank genes according to differential entropy between phenotypes, allowing key signalling genes associated with differentiation, metastasis or cancer to be identified \cite{Teschendorff2010bmc,West2012,Banerji2013}.

\section{Results}
\subsection{Signalling entropy and cellular robustness}
Our previous observation that signalling entropy is increased in cancer \cite{West2012}, and that one of the key characteristics of cancer cells is their robustness to intervention and environmental stresses, suggested to us that high cellular signalling entropy may be a defining feature of a cell that is robust to general perturbations. In addition, cancer cells often exhibit the phenomenon of oncogene addiction, whereby they become overly reliant on the activation of a specific oncogenic pathway, rendering them less robust to targeted intervention \cite{Hanahan2011}. Since oncogenes are normally overexpressed in cancer, it follows by the second cancer system-omic hallmark \cite{West2012}, that their lower signalling entropy may underpin their increased sensitivity to targeted drugs ({\bf Fig.2}). Based on these insights, we posited that cellular signalling entropy may be correlated to the overall cellular system's robustness to perturbations.\\
In order to explore the relation between signalling entropy and robustness in a general context, it is useful to consider another network property of the stochastic matrix, namely the global mixing rate. This mixing rate is defined formally as the inverse of the number of timesteps for which one has to evolve a random walk on the graph so that its probability distribution is close to the stationary distribution, independently of the starting node ({\bf Appendix B}). This is a notion that is closer to that of robustness or resilience as defined by Demetrius and Manke \cite{Demetrius2004,Demetrius2005,Manke2006}, allowing the mixing rate to be viewed as a crude proxy of a system's overall robustness to generic perturbations. Furthermore, the global mixing rate can be easily estimated as
\begin{equation}
\mu_R=-\log{\mathrm{SLEM}}
\end{equation}
where $SLEM$ is the second largest (right) eigenvalue modulus of the stochastic matrix ({\bf Appendix B}). Thus, we first asked how this global mixing rate is related to the signalling entropy rate.\\

\begin{figure}[ht]
\begin{center}
\includegraphics[scale=0.75]{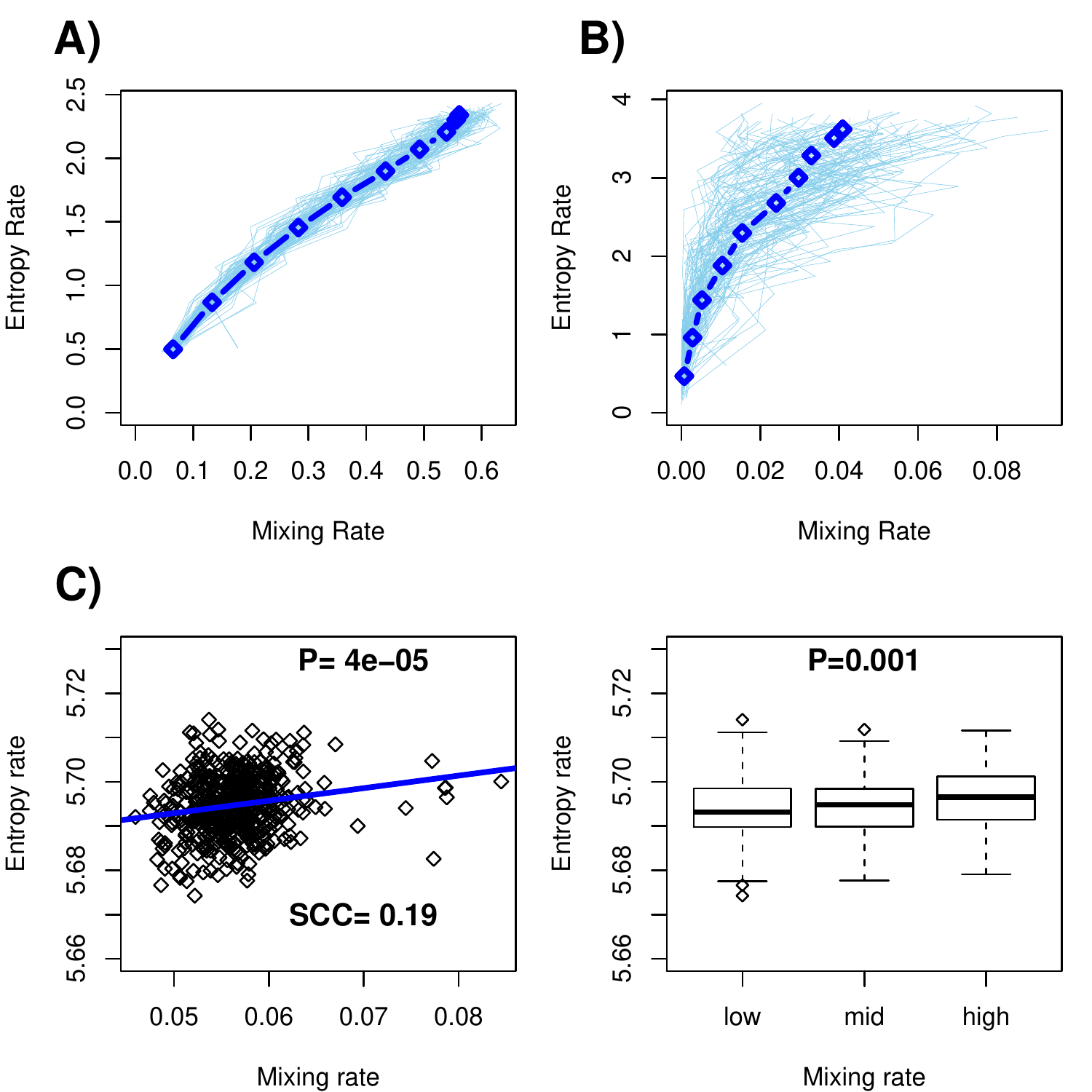}
%
%
\caption{{\bf Correlation between global entropy and mixing rates.} {\bf A)} Plotted is the entropy rate against the mixing rate for Erdos-Renyi graphs of 100 nodes and average degree 10. The light blue lines show the results over 100 different network realisations, with the dark blue line representing the ensemble average. {\bf B)} As A) but for connected subnetworks of average size 500 nodes, generated by random subsampling of 1000 nodes from the full PPI network of 8038 nodes. As in A), a range of edge weight distributions were considered reflecting variable entropy rates. The light blue lines show the results over 100 different realisations, with the dark blue line representing the ensemble average. {\bf C)} Scatterplot of the global entropy versus mixing rates for the 488 Cancer Cell-Line Encyclopedia (CCLE) samples. Spearman rank correlation coefficient (SCC) and associated P-value are given (left panel). Wilcoxon rank sum test P-value between high and low mixing rate groups (as defined by tertiles) (right panel).
}
\label{fig:4}       
\end{center}
\end{figure}

For a regular network of degree $d$ it can be easily shown that the entropy rate $SR=\log{d}$, whilst results on graph theory also show that for sufficiently large regular graphs, $\mu_R\propto \log{d}$ \cite{Friedman1991}. Hence, at least for regular networks a direct correlation exists. It follows that for non-regular graphs with tight degree distributions, e.g. Erd\"os-Renyi (ER) graphs, the entropy and mixing rates should also be approximately correlated. Indeed, using an empirical computational approach to evaluate the entropy and mixing rates for ER graphs with variable entropy rates, we were able to confirm this correlation ({\bf Appendix C, Fig.4A}). Next, we wanted to investigate if this relationship also holds for networks with more realistic topologies than ER graphs. Hence, we generated connected networks on the order of 500 nodes by random subsampling 1000 nodes from our large protein interaction network ($\sim 8000$ nodes) followed by extraction of the maximally connected component ({\bf Appendix D}). We verified that these networks possessed approximate scale-free topologies with clustering coefficients which were also significantly higher than for ER graphs. As before, for each generated network, stochastic matrices of variable entropy rates were defined. Signalling entropy and mixing rates were then estimated for each of these networks, and subsequently averaged over an ensemble of such graphs. As with the random Poisson (ER) graphs, average mixing and entropy rates were highly correlated ({\bf Fig.4B}).\\
Having demonstrated a direct correlation between these two measures on simulated data, we next asked if such a correlation could also be present in the full protein interaction networks and with the stochastic matrices derived from real expression data. Thus, we computed the global entropy and mixing rates for 488 cancer cell lines from the Cancer Cell Line Encylopedia (CCLE) ({\bf Appendix E}) \cite{Barretina2012}. Remarkably, despite the dynamic ranges of both entropy and mixing rates being much smaller ({\bf Fig.4C}) compared to those of the artificially induced stochastic matrices (c.f {\bf Fig.4A-B}), we were still able to detect a statistically significant correlation between entropy and mixing rates, as estimated across the 488 cell lines ({\bf Fig.4C}). Thus, all these results support the view that global entropy and mixing rates are correlated, albeit only in an average/ensemble sense.

\subsection{Local signalling entropy predicts drug sensitivity}
In the case of realistic expression and PPI data, the observed correlation between entropy and mixing rates was statistically significant but weak ({\bf Fig.4C}). This could be easily attributed to the fact that in real biological networks, the global mixing rate is a very poor measure of cellular robustness. In fact, it is well known that cellular robustness is highly sensitive to which genes undergo the perturbation. For instance, in mice some genes are embryonically lethal, whereas others are not \cite{Barabasi2004}. Robustness itself also admits many different definitions. Because of this, we decided to investigate signalling entropy in relation to other more objective measures of cellular robustness. One such measure is drug sensitivity, for instance, IC50 values, which measure the amount of drug dose required to inhibit cell proliferation by 50$\%$. According to this measure, a cell that is insensitive to drug treatment is highly robust to that particular treatment. Since most drugs target specific genes, we decided to explore the relation, if any, between the local signalling entropy of drug targets and their associated drug sensitivity measures. Specifically, we hypothesized that since local entropy provides a proxy for local pathway redundancy, that it would correlate with drug resistance. To test this, we first computed for each of the 8038 genes in the PPI network its local signalling entropy in each of the 488 CCLE cancer cell-lines. To identify associations between the 24 drug sensitivity profiles and the 8038 local network entropies, we computed non-linear rank (Spearman) correlation coefficients across the 488 cancer cell-lines, resulting in $24\times 8038$ correlations and associated P-values. We observed that there were many more significant associations than could be accounted for by random chance ({\bf Fig.5A}), with the overall strength of association very similar to that seen between gene expression and drug sensitivity ({\bf Fig.5B}).\\

\begin{figure}[ht]
\begin{center}
\includegraphics[scale=0.6]{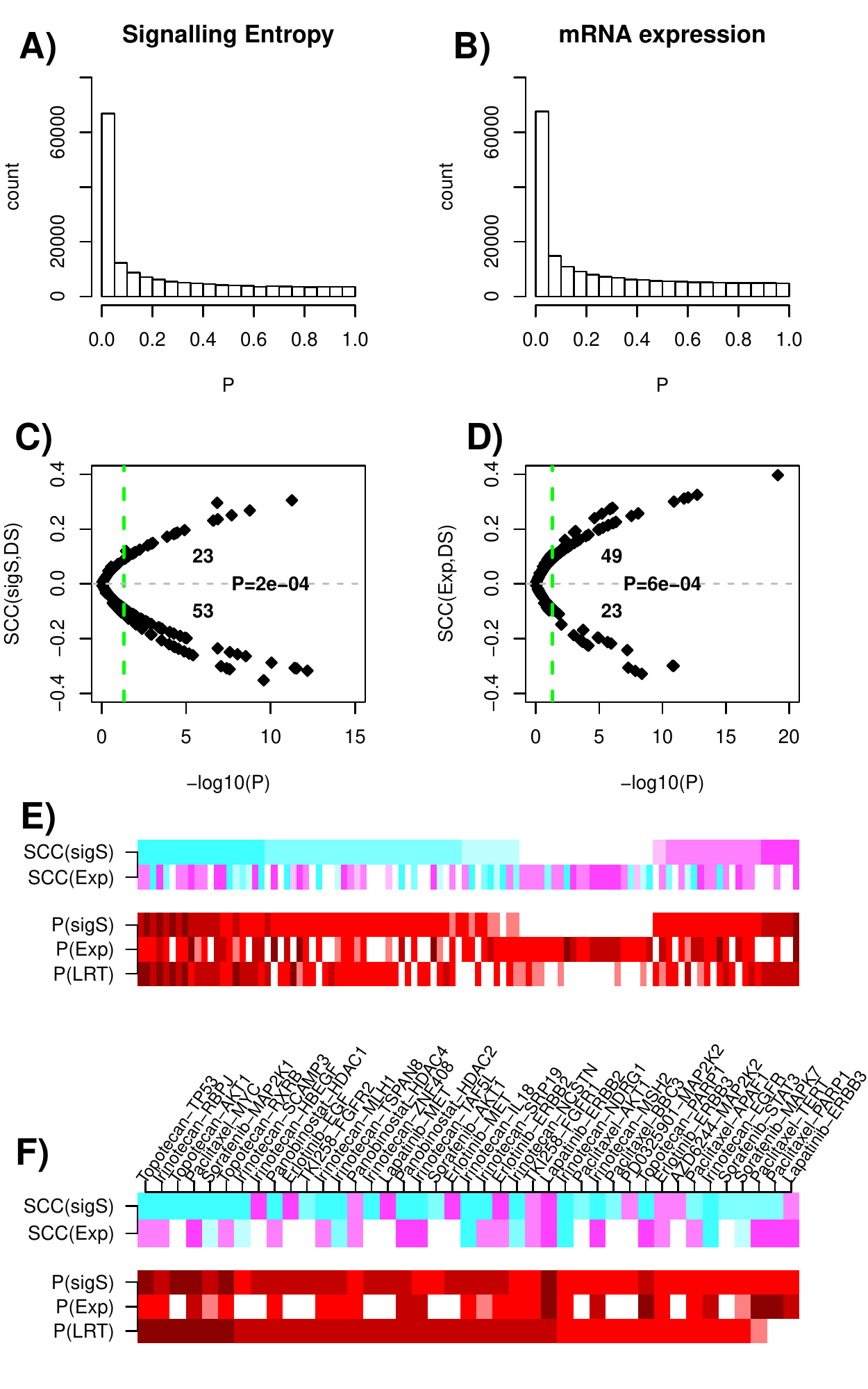}
%
%
\caption{{\bf Anti-correlation between local signalling entropy and drug sensitivity.} {\bf A)} Histogram of Spearman rank correlation P-values between drug sensitivities (n=24) and local signalling entropies (n=8038 genes), as computed over the 488 CCLE cell-lines. {\bf B)} As A) but for gene expression instead of signalling entropy. {\bf C)} Scatterplot of Spearman rank Correlation Coefficient (SCC) between local signalling entropy (sigS) and drug sensitivity (DS) against $-log_{10}$P-value for each of 134 drug gene target pairs. {\bf D)} As C) but for gene expression instead of local entropy. In {\bf C) \& D)}, we give the distribution of significant positive and negative correlations and the associated Binomial test P-value. {\bf E)} Drug target gene pairs ranked according to negative SCC (cyan color) between signalling entropy and drug sensitivity. Only pairs where at least one of entropy or gene expression were significantly associated are shown. Upper panels show the SCC (cyan=strong negative SCC, white=zero or non-significant SCC, magenta=strong positive SCC), while lower panels show the corresponding P-values with the last row showing the P-value from the Likelihood Ratio Test (LRT) assessing the added predictive value of signalling entropy over gene expression. The darker the tones of red the more significant the P-values are, whilst white indicates non-significance. {\bf F)} A subset of E), with pairs now ranked according to the LRT P-value.
}
\label{fig:5}       
\end{center}
\end{figure}

One would expect the targets of specific drugs to be highly informative of the sensitivity of the corresponding drugs. We used the CancerResource \cite{Ahmed2011} to identify targets of the 24 drugs considered here, and found a total of 154 drug-target pairs. For 134 of these pairs we could compute a P-value of association between local entropy and drug sensitivity with 76 pairs (i.e. 57$\%$) exhibiting a significant association ({\bf Fig.5C}). This was similar to the proportion (54$\%$) of observed significant associations between gene expression and drug sensitivity ({\bf Fig.5D}). However, interestingly, only 42 of these significant drug-target pairs were in common between the 76 obtained using signalling entropy and the 72 obtained using gene expression. Importantly, while the significant associations between gene expression and drug sensitivity involved preferentially positive correlations, in the case of signalling entropy most of the significant correlations were negative ({\bf Fig.5C-D}), exactly in line with our hypothesis that high entropy implies drug resistance. Thus, as expected, cell-lines with highly expressed drug targets were more sensitive to treatment by the corresponding drug, but equally, drug targets exhibiting a relatively low signalling entropy were also predictive of drug sensitivity.\\
To formally demonstrate that local signalling entropy adds predictive power over gene expression, we considered bi-variate regression models including both the target's gene expression as well as the target's signalling entropy. Using such bivariate models and likelihood ratio tests we found that in the majority of cases where signalling entropy was significantly associated with drug sensitivity that it did so independently of gene expression, adding predictive value ({\bf Fig.5E}). Top ranked drug-target pairs where signalling entropy added most predictive value over gene expression included Topotecan/{\it TP53} and Paclitaxel/{\it MYC} ({\bf Fig.5F}).\\
To further demonstrate that the observed associations between local signalling entropy and drug sensitivity are statistically meaningful, we conducted a control test, in which we replaced in the multivariate model the signalling entropy of the target with a non-local entropy computed by randomly replacing the PPI neighbours of the target with other ``far-away'' genes in the network. For a considerable fraction (41$\%$) of drug-target pairs, the original multivariate models including the local entropy constituted better predictive models than those with the non-local entropy (false discovery rate $<0.25$), indicating that the observed associations are driven, at least partly, by the network structure.

\subsection{High signalling entropy of intra-cellular signaling hubs is a hallmark of drug resistance}

\begin{figure}[ht]
\begin{center}
\includegraphics[scale=0.6]{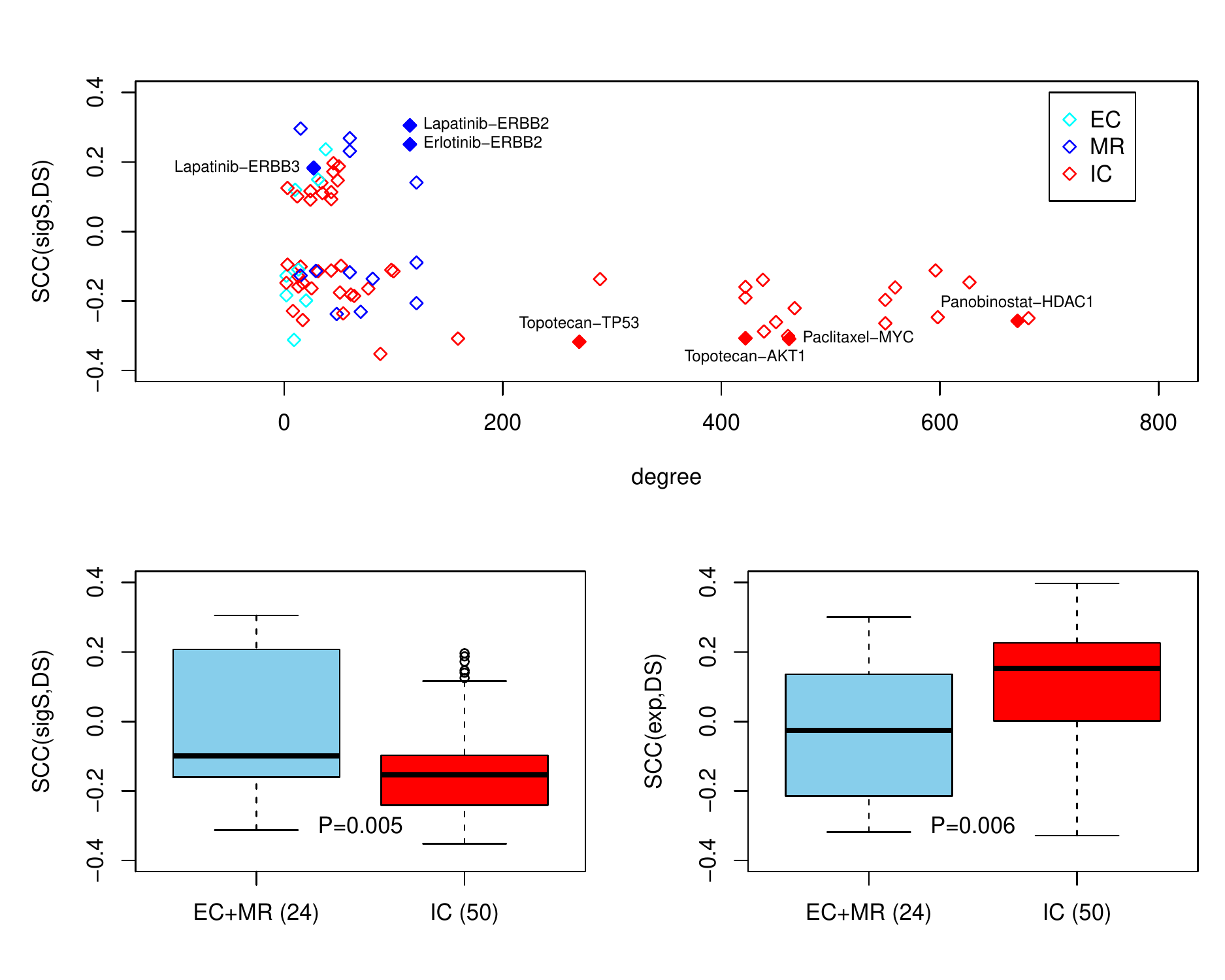}
%
%
\caption{{\bf High signalling entropy of intra-cellular hubs confers drug resistance:} Upper panel plots the topological degree of drug targets (x-axis) against the Spearman rank Correlation Coefficient (SCC) between its local signalling entropy and drug sensitivity, as assessed over the CCLE samples. EC=target annotated as extra-cellular, MR=target annotated as membrane receptor, IC=target annotated as intra-cellular. Left lower panel shows the difference in SCC values between the IC and EC+MR targets. Wilcoxon rank sum test P-value given. Right lower panel shows the difference in SCC values between the IC and EC+MR targets, where now the SCC value were computed between gene expression and drug sensitivity. Wilcoxon rank sum test P-value given.
}
\label{fig:6}       
\end{center}
\end{figure}

Among the drug-target pairs for which signalling entropy was a significant predictor of drug sensitivity, we observed a striking non-linear association with the topological degree of the targets in the network ({\bf Fig.6A}). In fact, for hubs in the network, most of which encode nodes located in the intracellular signaling hierarchy, high signalling entropy was exclusively associated with drug resistance (negative SCC). Examples included well-known intracellular signalling hubs like {\it HDAC1, HDAC2, AKT1, TP53, STAT3, MYC}). Some intracellular non-hubs (e.g. {\it CASP9, BCL2L1, BIRC3}) also exhibited negative correlations between signalling entropy and drug sensitivity. Among targets for which high signalling entropy predicted drug sensitivity, we observed several membrane receptors (e.g {\it ERBB2, ERBB3, EGFR, MET}) and growth factors (e.g {\it HBEGF, EGF, TGFA}). Given that the correlation coefficients were estimated across samples (cell-lines) and that the underlying network topology is unchanged between samples, the observed non-linear relation between the directionality of the correlation and node degree is a highly non-trivial finding. We also observed a clear dependence on the main signaling domain of the target, with intracellular hubs showing preferential anti-correlations, in contrast to growth factors and membrane receptors which exhibited positive and negative correlations in equal proportion ({\bf Fig.6B}). Thus, we can conclude from these analyses that cancer associated changes to the interaction patterns of intra-cellular hubs are key determinants of drug resistance. In particular, in cancers where the local entropy at these hubs is increased, as a result of increased promiscuous signaling, drugs targeting these hubs are less likely to be effective.

\section{Discussion and Conclusions}
Here we have advocated a fairly novel methodological framework, based on the notion of signalling entropy, to help analyze and interpret functional omic data sets. The method uses a structural network, i.e. a PPI network, from the outset, and integrates this with gene expression data, using local and global signalling entropy measures to estimate the amount of uncertainty in the network signaling patterns. We made the case as to why uncertainty or entropy might be a key concept to consider when analysing biological data. In previous work \cite{West2012,Banerji2013}, we showed how signalling entropy can be used to estimate the differentiation potential of a cellular sample in the context of normal differentiation processes, as well as demonstrating that signalling entropy also provides a highly accurate discriminator of cancer phenotypes.\\
In this study we focused on a novel application of signalling entropy to understanding cellular robustness in the context of cancer drug sensitivity screens. Our main findings are that (i) local signalling entropy measures add predictive value over models that only use gene expression, (ii) that the local entropy of drug targets generally correlates positively with drug resistance, and (iii) that increased local entropy of intra-cellular hubs in cancer cells is a key hallmark of drug resistance.\\
These results are consistent and suggestive of an underlying entropy-robustness correlation theorem, as envisaged by previous authors \cite{Demetrius2004}. Here, we provided additional empirical justification for such a theorem, using both simulated as well as real data, and using drug sensitivity measures as proxies for local robustness measures. A more detailed theoretical analysis of local mixing and entropy rates and incorporation of additional information (e.g. phosphorylation states of kinases, protein expression,..etc) when estimating entropies on real data, will undoubtedly lead to further improvements in our systems-level understanding of how entropy/uncertainty dictates cellular phenotypes. From a practical perspective, we have already shown in previous work \cite{Banerji2013} how local network entropies could be used to identify key signaling pathways in differentiation. It will therefore be interesting in future to apply the signalling entropy framework in more detail to identify specific signaling nodes/pathways underlying drug sensitivity/resistance.\\
Our results have also confirmed the importance of network topology (e.g. hubness and therefore scale-freeness) in dictating drug resistance patterns. Thus, it will be interesting to continue to explore the deep relation between topological features such as scale-freeness and hierarchical modularity in relation to the gene expression patterns seen in normal and disease physiology. It is entirely plausible that, although our current data and network models are only mere caricatures of the real biological networks, that underlying fundamental systems biology principles characterising cellular phenotypes can still be gleaned from this data. Indeed, our discovery that increased signalling entropy correlates with drug resistance demonstrates that such fundamental principles can already be inferred from existing data resources.\\
It will also be of interest to consider other potential applications of the signaling entropy method. For instance, one application could be to the identification of functional driver aberrations in cancer. This would first use epigenomic (e.g. DNA methylation and histone modification profiles) and genomic information (SNPs, CNVs) together with matched gene or protein expression data to identify functional epigenetic/genetic aberrations in cancer. Signalling entropy would subsequently be used as a means of identifying those aberrations which also cause a fundamental rewiring of the network. With multi-dimensional matched omic data readily available from the TCGA/ICGC, this represents another potentially important application of the signalling entropy method. Another important future application of the signaling entropy method would be to single-cell data, which is poised to become ever more important  \cite{Speicher2013}. So far, all signaling entropy analyses have been performed on cell populations, yet single-cell analysis will be necessary to disentangle the entropies at the single-cell and population-cell levels.\\
In summary, we envisage that signalling entropy will become a key concept in future analyses and interpretation of biological data.

\section*{Acknowledgements}
AET is supported by the Chinese Academy of Sciences, Shanghai Institute for Biological Sciences and the Max-Planck Gesellshaft. RK and PS acknowledge funding under FP7/2007-2013/grant agreement nr. 290038.

\appendix

\section{The protein protein interaction (PPI) network}
We downloaded the complete human protein interaction network from Pathway Commons ({\it www.pathwaycommons.org}) (Jun.2012) \cite{Cerami2011}, which brings together protein interactions from several distinct sources. We built a protein protein interaction (PPI) network from integrating the following sources: the Human Protein Reference Database (HPRD) \cite{Prasad2009}, the National Cancer Institute Nature Pathway Interaction Database (NCI-PID) ({\it pid.nci.nih.gov}), the Interactome (Intact) {\it http://www.ebi.ac.uk/intact/} and the Molecular Interaction Database (MINT) {\it http://mint.bio.uniroma2.it/mint/}. Protein interactions in this network include physical stable interactions such as those defining protein complexes, as well as transient interactions such as post-translational modifications and enzymatic reactions found in signal transduction pathways, including 20 highly curated immune and cancer signaling pathways from NetPath ({\it www.netpath.org}) \cite{Kandasamy2010}. We focused on non-redundant interactions, only included nodes with an Entrez gene ID annotation and focused on the maximally conntected component, resulting in a connected network of 10,720 nodes (unique Entrez IDs) and 152,889 documented interactions.

\section{Mixing rate results for a general random walk on a connected graph}
Suppose that we have a connected graph with an ergodic Markov chain defined on it, given by a stochastic matrix $P$ with stationary distribution $\pi$ (i.e. $\pi P=\pi$). We further assume that the detailed balance equation, $\pi_ip_{ij}=\pi_jp_{ji}$ holds. Defining the diagonal matrix $\Pi\equiv\mathrm{diag}(\pi_1,...,\pi_N)$, the detailed balance equation can be rewritten as $\Pi P=P^T\Pi$. That the matrix $P$ is stochastic means that each row of $P$ sums to 1. Equivalently, the vector with all unit entries, 1, is a right eigenvector of $P$, i.e. $P1=1$. Note also that the stationary distribution $\pi$ is a left eigenvector of $P$ with eigenvalue 1. The Perron-Frobenius theorem further implies that all other eigenvalues are less than 1 in magnitude. If detailed balance holds, all eigenvalues are also real.\\
The global mixing rate can be defined by considering the rate at which the node visitation probabilities of a random walker approaches that of the stationary distribution, independently of the starting position. Formally, if we let $Q_i(t)$ denote the probability that at time $t$ we find the walker at node $i$, then the mixing rate, $\mu_R$, is defined by \cite{Lovasz1993} $$\mu_R=\lim_{t\rightarrow\infty}\sup\max_i{|Q_i(t)-\pi_i|^{1/t}}.$$ Denoting by $Q(t)$ the column vector with elements $Q_i(t)$, one can write $$Q(t)=(P^t)^TQ(0).$$ To determine the elements, $Q_i(t)$, of this vector, it is convenient to first introduce the matrix $M=\Pi^{\frac{1}{2}}P\Pi^{-\frac{1}{2}}$. This is because $M^t=\Pi^{\frac{1}{2}}P^t\Pi^{-\frac{1}{2}}$, and so $P^t$ can be rewritten in terms of $M$, but also because $M$ satisfies the following lemma:
\newtheorem{bRW}{Lemma}
\begin{bRW}
$M$ has the same eigenvalues as $P$ and if $u_a$ is an eigenvector of $M$, then $r_a=\Pi^{-\frac{1}{2}}u_a$ and $l_a=\Pi^{\frac{1}{2}}u_a$ are right and left eigenvectors of $P$.
\end{bRW}
\begin{proof}[Proof of Lemma 1]
Suppose that $Mu_a=\lambda_au_a$. It then follows that $$P(\Pi^{-\frac{1}{2}}u_a)=\lambda_a(\Pi^{-\frac{1}{2}}u_a).$$ In the case of the left-eigenvector, multiply  $Mu_a=\lambda_au_a$ from the left with $\Pi^{\frac{1}{2}}$. Then, $$\Pi P\Pi^{-1}(\Pi^{\frac{1}{2}}u_a)=\lambda_a(\Pi^{\frac{1}{2}}u_a).$$ Detailed balance implies that $\Pi P=P^T\Pi$, so $P^Tl_a=\lambda_al_a$. Taking the transpose of this implies that $l_a$ is indeed a left-eigenvector of $P$.
\end{proof}
The significance of the above lemma becomes clear in light of the detailed balance equation, which implies that $M=M^T$, and so $M$ and $M^t$ can be orthogonally diagonalized. In particular, we can express $Q_i(t)$ in terms of the eigenvalue decomposition of $M^t$, as $$Q_i(t)=\sum_a{q_a|\lambda_a|^t\tilde{u}_{ai}\pi_i^{1/2}}$$, where $q_a=\sum_{j}{\tilde{u}_{aj}\pi_j^{-1/2}Q_j(0)}$ and where $\tilde{u}_a$ is the $a$'th unit norm eigenvector of $M$. Since $\Pi^{\frac{1}{2}}1$ is the top unit norm eigenvector of $M$ (using previous lemma), which has an eigenvalue of 1, it follows that $q_1=1$ and hence that 
$$Q_i(t)=\pi_i+\sum_{a\geq 2}{q_a|\lambda_a|^t\tilde{u}_{ai}\pi_i^{1/2}}.$$ It follows that $$|Q_i(t)-\pi_i|=\sum_{a\geq 2}{q_a|\lambda_a|^t\tilde{u}_{ai}\pi_i^{1/2}}.$$ Since $1=|\lambda_1|\geq|\lambda_2|\geq|\lambda_3|\dots$, we can conclude that as $t\rightarrow\infty$, the rate at which $Q_i$ approaches the stationary value $\pi_i$ is determined by the modulus of the second largest eigenvalue (the Second Largest Eigenvalue Modulus-SLEM). The global mixing rate $\mu_R$ can thus be estimated as $$\mu_R\approx -\log{|\lambda_2|}=-\log{\mathrm{SLEM}}$$.

\section{Entropy and mixing rates in simulated weighted Erd\"os-Renyi graphs}
For large regular graphs of degree $d$, the mixing rate is proportional to $\log{d}$ \cite{Friedman1991} and thus directly proportional to the entropy rate ($SR=\log{d}$ for a regular graph of degree $d$). By extrapolation, we can thus reasonably assume that for any sufficiently large graph with a tight degree distribution, such as random Erdos-Renyi (ER) graphs, that the entropy and mixing rates will also be correlated, albeit perhaps only in an average ensemble sense. The analytical demonstration of this is beyond the scope of this work. Hence, we adopt a computational empirical approach to see if the entropy and mixing rates may indeed be correlated. In detail (and without loss of generality concerning the end result) we considered ER graphs of size 100 nodes and average degree 10 (other values for these parameters gave similar answers). We built an ensemble of 100 distinct ER graphs, and for each of these we constructed a family of weighted networks, parameterised by a parameter $\epsilon$ which controls the level of signalling entropy. Specifically, we gradually shifted the weight distribution around each node $i$ (with degree $d_i$), in such a way that $p_{ij}=\epsilon/d_i$ for $j\neq k$ and $p_{ij}=1-\frac{\epsilon}{d_i}(d_i-1)$ for $j=k$, with $0\leq \epsilon \leq 1$ and with $k$ labeling a randomly chosen neighbor of node $i$. Thus, $\epsilon$ is a parameter that directly controls the uncertainty in the information flow, i.e. the entropy rate, since for $\epsilon=1$ we have that $p_{ij}=A_{ij}/d_i$ ($A_{ij}$ is the symmetric adjacency matrix of the graph), whilst for $\epsilon=0$, $p_{ij}=\delta_{ik}$, i.e. the information flow from node $i$ can only proceed along one node (node $k$). Thus, one would expect the entropy rate to decrease as $\epsilon$ is decreased to zero. For each value of $\epsilon$ we thus have an ensemble of 100 ER-graphs, for each of which the entropy and mixing rates can be computed.  Finally, at each value of $\epsilon$ the entropy and mixing rates are averaged over the ensemble. 

\section{Entropy and mixing rates in simulated weighted subgraphs of a PPI network}
The analysis described above was performed also for maximally connected subnetworks generated from the underlying realistic PPI network described earlier. Specifically, we randomly subsampled 1000 nodes from the 8038 node PPI network, and extracted the maximally connected subnetwork, which resulted (on average) in a subnetwork of approximately 500 nodes. A family of stochastic matrices of variable entropy rates were constructed as explained above and for each resulting weighted network we estimated the entropy and mixing rates. Finally, ensemble averages over 100 different realisations were computed.

\section{The Cancer Cell Line Encyclopedia (CCLE) data}
We used the gene expression data and drug sensitivity profiles as provided in the previous publication \cite{Barretina2012}. Briefly, integration of the gene expression data with our PPI network resulted in a maximally connected component consisting of 8038 genes/proteins. There were 488 cell-lines with gene expression and drug sensitivity profiles for 24 drugs. As a measure of drug response we used the Activity Area \cite{Barretina2012} since this measure gave optimal results when correlating drug response to gene expression levels of well established drug targets.




\bibliographystyle{elsarticle-num}




\section*{Figures}





\end{document}